\documentclass[runningheads,a4paper]{llncs}

\usepackage{enumerate}
\usepackage{graphicx}
\usepackage{amsmath,amssymb,graphicx} 
\usepackage{xcolor}

\usepackage{algpseudocode}

\algblockdefx[Loop]{Loop}{EndLoop}[1][]{\textbf{Loop} #1}{\textbf{End Loop}}

\usepackage[boxed,linesnumbered,noend]{algorithm2e}

\usepackage{tikz}
\usepackage{wrapfig}

\makeatletter
\providecommand{\leftsquigarrow}{%
  \mathrel{\mathpalette\reflect@squig\relax}%
}
\newcommand{\reflect@squig}[2]{%
  \reflectbox{$\m@th#1\rightsquigarrow$}%
}
\makeatother

\usepackage{url}
\urldef{\mailsa}\path|{ruchic, gburleigh}@ufl.edu|

\newcommand{\Prof}{\mathcal{P}}
\newcommand{\T}{\mathcal{T}}
\newcommand{\DE}{D^*}
\newcommand{\rt}{\mathrm{rt}}
\newcommand{\init}{\mathrm{init}}
\newcommand{\ter}{\mathrm{ter}}
\newcommand{\cnt}{\mathrm{count}}
\newcommand{\Desc}{\textsc{Descendant}}

\begin{document} \sloppy

\mainmatter  

\title{Constructing and Employing Tree Alignment Graphs for Phylogenetic Synthesis}

\titlerunning{Constructing and Employing Tree Alignment Graphs}

%
%
\author{Ruchi Chaudhary\inst{1}\and David Fern\'{a}ndez-Baca\inst{2}\and J. Gordon Burleigh\inst{1}}
%


\institute{Department of Biology, University of Florida, Gainesville, FL 32611, USA
 \and Department of Computer Science, Iowa State University, Ames, IA 50011, USA}

%
%

\maketitle

\newtheorem{observation}{\textbf{Observation}}

\begin{abstract}
Tree alignment graphs (TAGs) provide an intuitive data structure for storing phylogenetic trees that exhibits the relationships of the individual input trees and can potentially account for nested taxonomic relationships.   This paper provides a theoretical foundation for the use of TAGs in phylogenetics.  We provide a formal definition of TAG that --- unlike previous definition ---  does not depend on the order in which input trees are provided.  In the consensus case, when all input trees have the same leaf labels, we describe algorithms for constructing majority-rule and strict consensus trees using the TAG.  When the input trees do not have identical sets of leaf labels, we describe how to determine if the input trees are compatible and, if they are compatible, to construct a supertree that contains the input trees.
\end{abstract}

\section{Introduction}
Phylogenetic trees are graphs depicting the evolutionary relationships among species; thus, they are powerful tools for examining fundamental biological questions and understanding biodiversity (e.g., \cite{Baum:2012:K}). The wealth of available genetic sequences has rapidly increased the number of phylogenetic studies from across the tree of life (e.g., \cite{Goldman:2008:BS}). For example, STBase contains a million species trees generated from sequence data in GenBank \cite{McMahon:2015:PLoS}.  New next-generation sequencing technologies and sequence capture methods (e.g., \cite{Faircloth:2012:UEA,Lemmon:2012:AHE,McCormack:2012:UNP}) will further increase the rate in which phylogenetic data is generated in the coming years. This continuous flow of new phylogenetic data necessitates new approaches to store, evaluate, and synthesize existing phylogenetic trees.

Recently Smith et al.~introduced tree alignment graphs (TAGs) as a way to analyze large collections of phylogenetic trees \cite{Smith:PloSCB:2013}. TAGs preserve the structure of the input trees and thus  provide an intuitive, interpretable representation of the input trees, which enables users to visually assess patterns of agreement and conflict.  The TAG structure also makes it possible to combine trees whose tips include nested taxa (e.g., the tips of one tree contain species in taxonomic families, while the tips of another tree contain the families), which was true of only a few previous synthesis approaches \cite{Berry:2013:MLS,Berry:2006:RAB,Daniel:2005:AB}.
Indeed, a TAG was used to merge a taxonomy of all $\sim$2.3 million named species with $\sim$500 published phylogenetic trees to obtain an estimate of the tree of life \cite{Smith:2015:OT}.

The original TAG definition of Smith et al. \cite{Smith:PloSCB:2013} depends on the order of the input trees, which can be problematic.  Further, the several details of the synthesis process were not specified.    Our aim in this paper is to lay the theoretical foundations for further research on TAGs.  To this end, we first provide a mathematically precise definition of TAGs which is independent of the order of the input trees (Section \ref{definition}), and develop an algorithm for constructing TAGs (Section \ref{constTAG}).  We also describe algorithms to build strict and majority-rule consensus trees using TAGs (Section \ref{consensus}). We show how to check the compatibility among input trees and construct a supertree from compatible phylogenetic trees using a TAG (Section \ref{compatibility}).  Finally, we discuss the future applications and problems associated with using TAGs for assessing and synthesizing the enormous and rapidly growing number of available phylogenetic trees in the future (Section \ref{discussion}).

\paragraph{Related work.}
TAGs are part of a long history of using graph structures to synthesize the relationships among phylogenetic trees with partial taxonomic overlap.  The classic example is the \textsc{Build} algorithm \cite{Aho:1981:ITL,Semple:2003:phy} and its later variations (e.g., \cite{Berry:2006:RAB,Constantinescu:1995:AEA,Daniel:2005:AB,Ng:1996:ROR}).  These methods yield polynomial-time algorithms to determine whether a collection of input trees is compatible, and, if so, output the parent tree(s). Other graph-based algorithms, such as the \textsc{MinCutSupertree} \cite{Semple:2000:SMR}, the Modified \textsc{MinCutSupertree} \cite{Page:2002:MMS}, or the \textsc{MultiLevelSupertree} \cite{Berry:2013:MLS} algorithm allow users to synthesize collections of conflicting phylogenetic trees.
Although TAGs share important features with these earlier graph-theoretic approaches, TAGs display more directly the phylogenetic relationships exhibited by the input trees and therefore provide a more intuitive framework to examine patterns of conflict among trees \cite{Smith:PloSCB:2013}. TAGs also potentially summarize the information in the input trees with fewer nodes than previous graphs for semi-labeled trees \cite{Berry:2006:RAB,Daniel:2005:AB}.

\section{Preliminaries}
\label{definition}
\subsection{Notation}

Let  $T$ be a rooted tree.  Then, $\rt(T)$ and $\mathcal{L}(T)$ denote, respectively, the root and the leaf set of $T$, and $V(T)$ and $E(T)$ denote, respectively, the set of vertices and the set of edges of $T$. The set of all internal vertices of $T$ is $I(T) := V(T) \backslash \mathcal{L}(T)$. We define $\leq_{T}$ to be the partial order on $V(T)$ where $x \leq_{T} y$ if $y$ is a vertex on the path from $\rt(T)$ to $x$. If $\{x, y\} \in E(T)$ and $x \leq_{T} y$, then $y$ is the \emph{parent} of $x$ and $x$ is a \emph{child} of $y$. Two vertices in $T$ are \emph{siblings} if they share a parent. 

Let $X$ be a finite set of \emph{labels}. A \emph{phylogenetic tree on $X$} is a pair $\mathcal{T} = (T,\varphi)$ where 1) $T$ is a rooted tree in which every  internal vertex has degree at least three, except $\rt(T)$, which has degree at least two, and 2) $\varphi$ is a bijection from $\mathcal{L}(T)$ to $X$ \cite{Semple:2003:phy}. Tree $T$ is called the \emph{underlying tree} of $\mathcal{T}$ and $\varphi$ is called the \emph{labeling map} of $\mathcal{T}$. For convenience, we will often assume that the set of labels $X$ of $\mathcal{T}$ is simply $\mathcal{L}(\mathcal{T})$.
The \emph{size} of $\mathcal{T}$ is the cardinality of $\mathcal{L}(T)$.  $\mathcal{T}$ is \emph{binary} (or \emph{fully resolved}) if every vertex $v \in I(T) \setminus \rt(T)$ has degree three and $\rt(T)$ has degree two.

Let $\mathcal{T} = (T,\varphi)$ be a phylogenetic tree on $X$ and let $v$ be any vertex in $V(T)$. The subtree of $T$ rooted at vertex $v \in V(T)$, denoted by $T_v$, is the tree induced by $\{u \in V(T): u \leq v\}$. The \emph{cluster at $v$}, denoted $C_{\mathcal{T}}(v)$, is the set of leaf labels $\{\varphi(u) \in X: u \in \mathcal{L}(T_v)\}$.  We write $\mathcal{H}(\mathcal{T})$ to denote the set of all clusters of $\mathcal{T}$. Note that $\mathcal{H}(\mathcal{T})$ includes trivial clusters; i.e., clusters of size one or $|X|$. 

\subsection{Tree Alignment Graphs}
Here we define the tree alignment graph (TAG).  Our definition is somewhat different from that of Smith et al. \cite{Smith:PloSCB:2013}.  We explain these differences later.


We first need an auxiliary notion.  A \emph{directed multi-graph} is a directed graph that is allowed to have multiple edges between the same two vertices.  More formally, a directed multi-graph is a pair $(V,E)$ of disjoint sets (of vertices and edges) together with two maps $\init:E \rightarrow V$ and $\ter:E \rightarrow V$ assigning to each edge $e$ an \emph{initial vertex} $\init(e)$ and a \emph{terminal vertex} $\ter(e)$ \cite{Diestel:2000:graph}. Edge $e$ is said to be \emph{directed} from $\init(e)$ to $\ter(e)$.  The \emph{in-degree} of a node $v$ is the number of edges $e$ such that $\ter(e) = v$; the \emph{out-degree} of $v$ is the number of edges $e$ such that $\init(e) = v$.  We call a node with out-degree zero a \emph{leaf node}; all non-leaf nodes are called \emph{internal nodes}.

A \emph{directed acyclic (multi-) graph}, \emph{DAG} for short, is a directed multi-graph with no cycles.

\begin{definition} [Tree Alignment Graph (TAG)]\label{dfn:TAG} Let $\mathcal{P}$ be a collection of phylogenetic trees and $S = \bigcup_{\mathcal{T} \in \mathcal{P}} \mathcal{L}(\mathcal{T})$. The \emph{tree alignment graph} of $\mathcal{P}$ is a directed graph  $D = (U,E)$ along with an injective function $f: U \rightarrow 2^S$,  called the \emph{vertex-labeling function}, such that
\begin{itemize}
    \item for each $v \in U$, $f(v) \in \mathcal{H}(\mathcal{T})$ for some $\mathcal{T} \in \mathcal{P}$, and
    \item for each $\mathcal{T}:=(T,\varphi) \in \mathcal{P}$ and each $e := \{x,y\} \in E(T)$  where $x<_{T} y$, there exists a unique $e' \in E$ such that $C_{\mathcal{T}}(x) = f(\ter(e'))$ and $C_{\mathcal{T}}(y) = f(\init(e'))$.
\end{itemize}
\end{definition}

Figure \ref{fig:mytag} illustrates Definition \ref{dfn:TAG}.
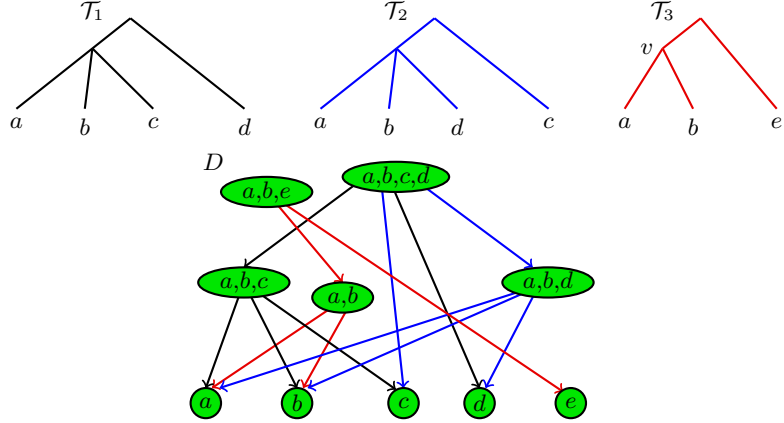
\begin{figure}
    \centering
     \begin{tikzpicture}
        \draw [thick] (0,1.2) to (-0.5,0.8);
        \draw [thick] (0,1.2) to (1.5,0) node [below] {\textbf{$d$}};
        \draw [thick] (-0.5,0.8) to (0.3,0) node [below] {\textbf{$c$}};
        \draw [thick] (-0.5,0.8) to (-0.6,0) node [below] {\textbf{$b$}};
        \draw [thick] (-0.5,0.8) to (-1.5,0) node [below] {\textbf{$a$}};
        \node at (-0.5,1.3) {$\mathcal{T}_1$};

        \draw [thick,blue] (4,1.2) to (3.5,0.8);
        \draw [thick,blue] (4,1.2) to (5.5,0) node [black,below] {\textbf{$c$}};
        \draw [thick,blue] (3.5,0.8) to (4.3,0) node [black,below] {\textbf{$d$}};
        \draw [thick,blue] (3.5,0.8) to (3.4,0) node [black,below] {\textbf{$b$}};
        \draw [thick,blue] (3.5,0.8) to (2.5,0) node [black,below] {\textbf{$a$}};
        \node at (3.5,1.3) {$\mathcal{T}_2$};

        \draw [thick,red!90!black] (7.5,1.2) to (7,0.8) node [black,left] {$v$};
        \draw [thick,red!90!black] (7.5,1.2) to (8.5,0) node [black,below] {\textbf{$e$}};
        \draw [thick,red!90!black] (7,0.8) to (7.4,0) node [black,below] {\textbf{$b$}};
        \draw [thick,red!90!black] (7,0.8) to (6.5,0) node [black,below] {\textbf{$a$}};
        \node at (7,1.3) {$\mathcal{T}_3$};
    \end{tikzpicture}

     \begin{tikzpicture}

        \draw [thick,->,red!90!black] (7.5,1.5) to (5.58,0.2);  
        \draw [thick,->,red!90!black] (7.5,1.5) to (6.78,0.2);  

        \draw [thick,->,black] (7.6,3) to (6,1.8);  
        \draw [thick,->,black] (6,1.6) to (5.5,0.2);  
        \draw [thick,->,black] (6,1.6) to (6.7,0.2);  
        \draw [thick,->,black] (6,1.6) to (8,0.18);  
        \draw [thick,->,black] (7.9,3) to (9.1,0.2);  

        \draw [thick,->,blue] (8.2,3) to (9.8,1.8);  
        \draw [thick,->,blue] (10,1.6) to (5.68,0.2);  
        \draw [thick,->,blue] (10,1.6) to (6.84,0.2);  
        \draw [thick,->,blue] (9.9,1.6) to (9.18,0.2);  
        \draw [thick,->,blue] (7.8,3) to (8.1,0.2);  

        \draw [thick,->,red!90!black] (6.3,2.8) to (7.3,1.6);  
        \draw [thick,->,red!90!black] (6.3,2.8) to (10.2,0.15);  

        \draw [thick,fill=green!90!black] (6,1.6) ellipse (6mm and 2mm);
        \node at (6,1.6) {$a$,$b$,$c$};

        \draw [thick,fill=green!90!black] (7.3,1.4) ellipse (4mm and 2mm);
        \node at (7.3,1.4) {$a$,$b$};

        \draw [thick,fill=green!90!black] (6.3,2.8) ellipse (6mm and 2mm);
        \node at (6.3,2.8) {$a$,$b$,$e$};

        \draw [thick,fill=green!90!black] (8,3) ellipse (7mm and 2mm);
        \node at (8,3) {$a$,$b$,$c$,$d$};
        \node at (5.6,3.2) {$D$};

        \draw [thick,fill=green!90!black] (10,1.6) ellipse (6mm and 2mm);
        \node at (10,1.6) {$a$,$b$,$d$};

        \draw [thick,fill=green!90!black] (5.5,0) circle [radius=0.2];
        \draw [thick,fill=green!90!black] (6.7,0) circle [radius=0.2];
        \draw [thick,fill=green!90!black] (8.1,0) circle [radius=0.2];
        \draw [thick,fill=green!90!black] (9.1,0) circle [radius=0.2];
        \draw [thick,fill=green!90!black] (10.3,0) circle [radius=0.2];

        \node at (5.5,0) {$a$};
        \node at (6.7,0) {$b$};
        \node at (8.1,0) {$c$};
        \node at (9.1,0) {$d$};
        \node at (10.3,0) {$e$};
      \end{tikzpicture}
    \vspace{-0.1cm}
    \caption{A collection $\mathcal{P}$ of phylogenetic trees $\mathcal{T}_1$, $\mathcal{T}_2$, and $\mathcal{T}_3$ and the TAG $D$ of $\mathcal{P}$. The edges of $D$ are colored so as to correspond to the input trees; the internal vertices are labeled for clarity.}\label{fig:mytag}
    \vspace{-0.2cm}
  \end{figure}

\paragraph{Remarks:}
 \begin{enumerate}
   \item Note that we only use the vertex-labeling function $f$ to facilitate the definition and to label the leaf nodes of the TAG. We do not actually label the internal vertices of the TAG, since, as the TAG gets bigger, assigning labels using $f$ becomes impractical.
   \item Having a unique edge in the TAG for each input tree edge enables systematically annotating the TAG for each individual input tree, as the TAG also provides a means for storing phylogenetic trees.  
 \end{enumerate}

\begin{lemma} \label{lemma1} The TAG is acyclic. \end{lemma}
\begin{proof} Stems from the fact that for each edge $e \in E$, $f(\ter(e)) \subset f(\init(e))$. \qed \end{proof}

\paragraph{Comparison with Smith et al.'s TAG.} 
In  \cite{Smith:PloSCB:2013}, Smith et al.~define their TAG procedurally, as follows.  Let $\mathcal{P}$ be a collection of phylogenetic trees and $S = \bigcup_{\T\in \Prof} \mathcal{L}(\T)$. Let $D = (U,E)$ be a directed graph along with an injective vertex-labeling function $f: U \rightarrow 2^S$. Initially, $D$ has a vertex $f^{-1}(S)$, and $|S|$ vertices $f^{-1}(\{s\})$, for each $s \in S$.  Next, Smith et al.'s method process each input tree $\mathcal{T} = (T,\varphi) \in \mathcal{P}$, in some order, and does the following:

\begin{enumerate}
        \item Map each vertex $v \in \mathcal{L}(T)$ to vertex $u \in U$ where $\varphi(v) = f(u)$. 
        \item \label{step:add} Map each vertex $v \in I(T)$ to the vertex $u \in U$, where $C_{\mathcal{T}}(v) \cap f(u) \neq \phi$, $\mathcal{L}(\mathcal{T}) \setminus C_{\mathcal{T}}(v) \cap f(u) = \phi$, and $C_{\mathcal{T}}(v) \cap S \setminus f(u) = \phi$. If no such $u$ exists, then add new vertex $u$ with $f(u) := C_{\mathcal{T}}(v)$ in $D$.
        \item In the case of a vertex $v \in I(T)$ mapping to multiple vertices in $D$, for each such $t$ vertices $u_1,...,u_t \in U$, where for each $j \in \{1, \dots, t-1\}$ there exists $e' \in E$ such that $\ter(e')=u_j$ and $\init(e')=u_{j+1}$, discard all mapping of $v$ to $u_2,...,u_t$, except $u_1$. Note that $v \in V(T)$ can still be mapped to multiple vertices in $D$. For example, vertex $v$ of $T_3$ in Fig. \ref{fig:mytag} is mapped to vertices $f^{-1}(\{a,b,c\})$ and $f^{-1}(\{a,b,d\})$ of $D_1$ in Fig. \ref{fig:cont}.
        \item For edge $e = \{x,y\} \in E(T)$, add directed edges in $D$ from all mappings of $x$ to all mappings of $y$.
\end{enumerate}

Observe that Smith et al.'s definition of the TAG coincides with Definition \ref{dfn:TAG} when the input trees have completely overlapping leaf label sets; however, it differs when the input trees have partially overlapping leaf label sets, as we discuss next.

Notice that Step~\ref{step:add} first tries to map a vertex $v \in I(T)$ to a vertex $u \in U$ for which $C_{\mathcal{T}}(v) \subseteq f(u)$, and $f(u)$ does not have any label of $\mathcal{L}(T)$ other than that of $C_{\mathcal{T}}(v)$. If no suitable match exists, then a new vertex $f^{-1}(C_{\mathcal{T}}(v))$ is added to $D$.  As a result, the set of vertices in the TAG that this procedure creates can depend on the order of input trees. Fig. \ref{fig:cont} illustrates how changing the order of input trees can lead to different TAGs for the input trees of Fig. \ref{fig:mytag}. When $\mathcal{T}_3$ is processed after $\mathcal{T}_1$ and $\mathcal{T}_2$, the vertex $v$ of $\mathcal{T}_3$ maps to $f^{-1}(\{a,b,c\})$ and $f^{-1}(\{a,b,d\})$ in Step \ref{step:add}. On the other hand, processing $\mathcal{T}_3$ before $\mathcal{T}_1$ and $\mathcal{T}_2$, necessitates the creation of $f^{-1}(\{a,b\})$ in the resulting TAG.

Smith et al. discussed the possibility of order dependence of their TAG and addressed it through a post-processing procedure \cite{Smith:PloSCB:2013}. For the given collection of input trees and the TAG that results from the first round of processing, the post-processing procedure recomputes the mapping of each internal vertex of the input tree following Step 2. If the new mapping of an internal vertex of the input tree differs from the old mapping, then the mapping is updated. Edges of the resulting TAG that correspond to the outdated mapping are removed and edges for the new mapping are added. For example, there will be no change in TAG $D_1$ after applying post-processing procedure. On the contrary, post-processing will map $v \in I(\mathcal{T}_3)$ to $f^{-1}(\{a,b,c\})$ and $f^{-1}(\{a,b,d\})$ in $D_2$. This new mapping will cause adding directed edges, 1) from $f^{-1}(\{a,b,c,d,e\})$ to $f^{-1}(\{a,b,c\})$ and $f^{-1}(\{a,b,d\})$, 2) from $f^{-1}(\{a,b,c\})$ to $f^{-1}(\{a\})$ and $f^{-1}(\{b\})$, and 3) from $f^{-1}(\{a,b,d\})$ to $f^{-1}(\{a\})$ and $f^{-1}(\{b\})$ in $D_2$. Let $D_2'$ be the resulting TAG after applying post-processing on $D_2$. Clearly, $D_1$ and $D_2'$ are different. We note that since the post-processing is inadequate for overcoming order-dependence, an algorithm for pre-processing of input trees is in development\footnote{S. A. Smith, J. W. Brown, and C. E. Hinchliff (Department of Ecology and Evolutionary Biology, University of Michigan, Ann Arbor), personal communication.}. In contrast to Smith et al.'s TAG \cite{Smith:PloSCB:2013}, our TAG (in Definition 1) is independent of the order of input trees.

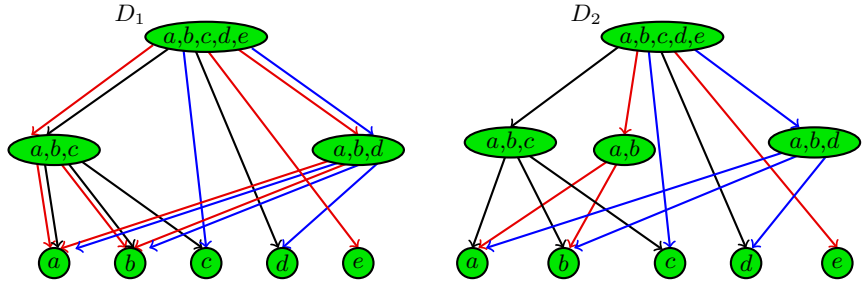
\begin{figure}
    \centering
    \begin{tikzpicture}

        \draw [thick,->,red!90!black] (1.3,2.89) to (-0.3,1.7);  
        \draw [thick,->] (1.5,2.84) to (-0.1,1.7);

        \draw [thick,->,red!90!black] (-0.25,1.5) to (-0.05,0.2);  
        \draw [thick,->] (-.15,1.5) to (0.05,0.2);

        \draw [thick,->,red!90!black] (-0.05,1.5) to (0.95,0.2);  
        \draw [thick,->] (.05,1.5) to (1.05,0.2);

        \draw [thick,->] (.15,1.5) to (1.95,0.2);   
        \draw [thick,->] (1.85,2.84) to (2.95,0.2);   
        \draw [thick,->,blue] (1.7,2.84) to (2,0.2);   

        \draw [thick,->,red!90!black] (2.4,2.85) to (4,1.7); 
        \draw [thick,->,blue!] (2.6,2.85) to (4.2,1.7);

        \draw [thick,->,red!90!black] (4.05,1.5) to (0.08,0.19);  
        \draw [thick,->,blue] (4.25,1.5) to (0.28,0.19);

        \draw [thick,->,red!90!black] (4.3,1.5) to (1.05,0.2);  
        \draw [thick,->,blue] (4.5,1.5) to (1.25,0.2);

        \draw [thick,->,blue] (4.45,1.5) to (3,0.2);  

        \draw [thick,->,red!90!black] (2,2.84) to (4,0.2); 

        \draw [thick,fill=green!90!black] (0,1.5) ellipse (6mm and 2mm);
        \node at (0,1.5) {$a$,$b$,$c$};

        \draw [thick,fill=green!90!black] (2,3) ellipse (8mm and 2mm);
        \node at (2,3) {$a$,$b$,$c$,$d$,$e$};
        \node at (1,3.3) {$D_1$};

        \draw [thick,fill=green!90!black] (4,1.5) ellipse (6mm and 2mm);
        \node at (4,1.5) {$a$,$b$,$d$};

        \draw [thick,fill=green!90!black] (0,0) circle [radius=0.2];
        \draw [thick,fill=green!90!black] (1,0) circle [radius=0.2];
        \draw [thick,fill=green!90!black] (2,0) circle [radius=0.2];
        \draw [thick,fill=green!90!black] (3,0) circle [radius=0.2];
        \draw [thick,fill=green!90!black] (4,0) circle [radius=0.2];

        \node at (0,0) {$a$};
        \node at (1,0) {$b$};
        \node at (2,0) {$c$};
        \node at (3,0) {$d$};
        \node at (4,0) {$e$};

        \draw [thick,->,red!90!black] (7.7,3) to (7.5,1.7);  
        \draw [thick,->,red!90!black] (7.5,1.5) to (5.58,0.2);  
        \draw [thick,->,red!90!black] (7.5,1.5) to (6.78,0.2);  
        \draw [thick,->,red!90!black] (8,3) to (10.3,0.2);  

        \draw [thick,->,black] (7.6,3) to (6,1.8);  
        \draw [thick,->,black] (6,1.6) to (5.5,0.2);  
        \draw [thick,->,black] (6,1.6) to (6.7,0.2);  
        \draw [thick,->,black] (6,1.6) to (8,0.18);  
        \draw [thick,->,black] (7.9,3) to (9.1,0.2);  

        \draw [thick,->,blue] (8.2,3) to (9.8,1.8);  
        \draw [thick,->,blue] (10,1.6) to (5.68,0.2);  
        \draw [thick,->,blue] (10.2,1.6) to (6.84,0.2);  
        \draw [thick,->,blue] (10.3,1.6) to (9.18,0.2);  
        \draw [thick,->,blue] (7.8,3) to (8.1,0.2);  

        \draw [thick,fill=green!90!black] (6,1.6) ellipse (6mm and 2mm);
        \node at (6,1.6) {$a$,$b$,$c$};

        \draw [thick,fill=green!90!black] (7.5,1.5) ellipse (4mm and 2mm);
        \node at (7.5,1.5) {$a$,$b$};

        \draw [thick,fill=green!90!black] (8,3) ellipse (8mm and 2mm);
        \node at (8,3) {$a$,$b$,$c$,$d$,$e$};
        \node at (7,3.3) {$D_2$};

        \draw [thick,fill=green!90!black] (10,1.6) ellipse (6mm and 2mm);
        \node at (10,1.6) {$a$,$b$,$d$};

        \draw [thick,fill=green!90!black] (5.5,0) circle [radius=0.2];
        \draw [thick,fill=green!90!black] (6.7,0) circle [radius=0.2];
        \draw [thick,fill=green!90!black] (8.1,0) circle [radius=0.2];
        \draw [thick,fill=green!90!black] (9.1,0) circle [radius=0.2];
        \draw [thick,fill=green!90!black] (10.3,0) circle [radius=0.2];

        \node at (5.5,0) {$a$};
        \node at (6.7,0) {$b$};
        \node at (8.1,0) {$c$};
        \node at (9.1,0) {$d$};
        \node at (10.3,0) {$e$};
    \end{tikzpicture}
  \caption{Following \cite{Smith:PloSCB:2013}, if the input trees from Fig. \ref{fig:mytag} are processed in the order $\mathcal{T}_1$, $\mathcal{T}_2$, and $\mathcal{T}_3$, the resulting TAG is $D_1$; changing the order of input trees, and processing them as $\mathcal{T}_3$, $\mathcal{T}_1$, and $\mathcal{T}_2$, results into $D_2$, which is different from $D_1$. Again, the edges of both TAGs are colored so to correspond to the input trees; the internal vertices are labeled for clarity.}\label{fig:cont}
  \end{figure} 

The input trees in \cite{Smith:PloSCB:2013} always include a taxonomy tree, which contains all of the leaf labels from the input trees and could be a star tree (i.e., all the leaf nodes connecting to a root node) if no taxonomic classification is available. Here we study the TAG in the standard supertree framework, so we do not assume that a taxonomy tree is included.

\section{Constructing a TAG}
\label{constTAG}
We now present an algorithm for constructing the TAG for a collection of phylogenetic trees. The algorithm first collects clusters by reading through the input trees and making each unique cluster a node in the TAG. It then adds edges between nodes in the TAG.

Let $\mathcal{P}$ be the input collection of $k$ phylogenetic trees. Let $S = \bigcup_{\mathcal{T}\in \mathcal{P}} \mathcal{L}(\mathcal{T})$ and $n = |S|$.

\subsection{Building TAG Nodes}
\label{nodes}
We define a bijection $g$
that maps each taxon in $S$ to a unique number in $\{1,2, \dots ,n\}$. For each tree $\mathcal{T} = (T,\varphi)$ in $\mathcal{P}$, the \emph{bit-string} of $v \in V(T)$ is a  binary string of length $n$, where the $i$th bit is 1, if $g^{-1}(i) \in C_{\mathcal{T}}(v)$, or 0, otherwise. We collect bit-strings from the input trees in a list $\mathcal{A}$ and construct a TAG node for each unique bit-string.

\paragraph{Collecting bit-strings:} The algorithm starts by traversing each input tree in post-order.  When,  after  traversing its subtree, a vertex $v$ is visited, we compute $v$'s bit-string as follows.  If $v$ is a leaf with label $s \in S$, the bit string for $v$ is simply the string of length $n$ with a 1 at position $g(s)$ and 0s everywhere else.
If $v$ is an internal node, its bit-string is the \emph{OR} of the bit-strings of $v$'s children. After each bit-string is computed, it is stored in $\mathcal{A}$.  When the traversals of all $k$ input trees are complete, $\mathcal{A}$ has $O(nk)$ bit-strings. 

\paragraph{Filtering unique bit-strings:} To remove duplicates from $\mathcal{A}$, we first sort it using \emph{radix sort} \cite[Chapter 8]{Cormen:1996:IA}.  Given $N$  $b$-bit numbers and any positive integer $r \le b$, radix sort sorts these numbers in $O((b/r)(N + 2^r))$ time.  In our case, $b = n$ and $N = nk$, giving a running time of $O((n/r)(nk + 2^r))$.  This quantity is minimized when $r = \log (nk)$, giving a running time of $O(n^2 k/\log(nk))$.

After sorting $\mathcal{A}$, we remove its duplicate bit-strings in a single linear scan.  This can be done in $O(n^2 k)$ time through standard methods.

\paragraph{TAG nodes:}  We construct a node in the TAG for each unique bit-string in $\mathcal{A}$. For bit-strings corresponding to the leaf nodes, we also associate the appropriate label from $S$ using function $g$.

\subsection{Adding Edges to the TAG}
\label{edges}
Once the vertices of input trees have been mapped to the vertices of TAG, we add directed edges to the TAG. We traverse each input tree in post-order. When a tree traversal visits an internal vertex $v$ of an input tree, we find the bit-strings of $v$ and $v$'s children in $\mathcal{A}$ and locate the nodes corresponding to them in the TAG. We then add directed edges from the TAG node corresponding to $v$ to the TAG nodes for $v$'s children.

\begin{theorem} For a given collection of $k$ phylogenetic trees on $n$ labels, the TAG can be built in $O(n^2k)$ time. \end{theorem}
\begin{proof} Collecting  bit-strings 
and then sorting them 
requires $O(n^2k)$ time.
The remaining steps take time linear in the size of the input. \qed \end{proof}

\section{Finding Consensus Trees using the TAG}
\label{consensus}
Let $\mathcal{P}$ be a collection of $k$ input phylogenetic trees with completely overlapping leaf label set of size $n$.
The \emph{strict consensus tree} for $\mathcal{P}$ is the tree whose clusters are precisely those that appear in all the trees in $\mathcal{P}$.  The \emph{majority-rule consensus tree} for $\mathcal{P}$ is the tree whose clusters are precisely those that appear in more than half (i.e., the majority) of the trees in $\Prof$.
Here we show how to build the majority-rule consensus trees for $\Prof$ from the TAG for $\Prof$.  We then outline the modifications needed to compute the strict consensus tree.

Algorithm \textsc{MajorityRuleTree} (Algorithm~\ref{alg:Majority}) builds the majority-rule tree for $\Prof$ by traversing the TAG $D$ for $\mathcal{P}$.  Let $D  = (U,E)$ and $f$ be the vertex-labeling function.  We assume that each vertex $v$ in $D$ stores the \emph{cardinality} of $v$ --- i.e., the number of taxa in $f(v)$ --- along with  $\cnt(v)$, the number of times cluster $f(v)$ appears in a tree in $\Prof$. We also assume that multiple edges between the same two vertices are replaced by a single edge. The next observation follows from the fact that the input trees have completely overlapping leaf label sets.

\begin{observation}\label{jobs:rootD}
$D$ has precisely one vertex $s$ with in-degree zero.
\end{observation}

Let $v$ be a node in $D$.  Then, $v$ is a \emph{majority node} if $\cnt(v) > k/2$.  The clusters associated with majority nodes  are precisely the clusters of the majority-rule tree. Let the nodes of the majority-rule tree correspond to the majority nodes of $D$. Next we develop an approach to hook up these nodes to actually build the majority-rule tree.

Let $u$ and $v$ be nodes of $D$.  Then, $v$ is a \emph{majority ancestor} of $u$ if $v$ is a majority node, and there is directed path from $v$ to $u$ in $D$.
Algorithm~\ref{alg:Majority} is based on the following observation (parts of which were noted in \cite{Amenta:2003:ALM}).

\begin{observation}
\label{obs:majObs}
Let $u$ and $v$ be majority nodes in $D$.  Then,
\begin{enumerate}[(i)]
\item
if there is a directed path from a majority node $u$ to a majority node $v$ in $D$, then $f(v) \subset f(u)$, and
\item
\label{obs:pt1}
if $v$ is the parent of $u$ in the majority-rule tree for $\Prof$, then
$v$ is the (unique) minimum-cardinality majority ancestor of $u$; further, $f(u) \subset f(v)$.
\end{enumerate}
\end{observation}

Let $u$ be a node in $D$.  The \emph{most recent majority ancestor of $u$} is the unique minimum-cardinality majority ancestor $u$.   For each vertex $u \in U$, our algorithm maintains two variables: $p(u)$, a reference to the smallest cardinality majority ancestor of $u$ seen thus far, and $m(u)$, the cardinality of $p(u)$.  Initially, every node $u$, except the node $s$ of in-degree zero, has $p(u) = s$, representing initial best estimate of the most recent majority ancestor of $u$.  The algorithm revises this estimate repeatedly until it converges on the correct value.  After this process is complete, it is now a simple matter to assemble the majority-rule tree, since, for each majority node $u$, $p(u)$ points to $u$'s parent in that tree.

Algorithm \ref{alg:Majority} processes the nodes of $D$ according to topological order --- this ordering exists because $D$ is acyclic (from Lemma 1).  When the algorithm visits a node $u$, it examines each successor $v$, and considers two possibilities.  If $u$ is a majority node, then $u$ may become the new value of $p(v)$, while if $u$ is not a majority node, $p(u)$ may become the new value of $p(v)$.  By Observation~\ref{obs:majObs}, the decision depends solely on node cardinalities.

\begin{figure}
\vspace{-0.2cm}
\begin{algorithm}[H]
\SetAlgoLined
\SetNoFillComment
\DontPrintSemicolon
\KwIn{The TAG $D = (U,E)$ for a collection $\mathcal{P}$ of trees over the same leaf set.}
\KwOut{The majority-rule tree for $\Prof$.}
Let $s$ be the unique vertex in $D$ with in-degree $0$\;
\ForEach{$u \in V -s$}{
	$m(u) = n$ ;
	$p(u) = s$
}
Perform a topological sort of $D-s$ \;
Let $M \subseteq U$ be the set of majority nodes in $D$ \;
\ForEach{$u \in U -s $ in topological order\label{step:mainLoop}}{
	\If {$u \in M$}{
		$\mu = |f(u)|$ ;
		$\pi = u$
	}\Else{
		$\mu = m(u)$ ;
		$\pi = p(u)$
	}
	\ForEach{$v \in U$ such that $(u,v) \in E$}{
		\If{$m(v) > \mu$}{
			$m(v) = \mu$ ;
			$p(v) = \pi$ \label{step:endMainLoop}
	}
	}
}
Let $T$ be the tree with vertex set $M$, where, for every $u \in M$, the parent of $u$ in $T$ is $p(u)$ \;
Let $\varphi$ be the function that maps each leaf $u$ of $T$ to $f(u)$ \;
\Return $(T, \varphi)$ \;
\label{alg:Majority}
\caption{\textsc{MajorityRuleTree}($D$)\label{alg:Majority}}
\end{algorithm}
\vspace{-0.2cm}
\end{figure}

\begin{theorem} Given the TAG $D$ for a collection $\Prof$ of $k$ phylogenetic trees on the same $n$ leaves, the majority-rule consensus tree of $\Prof$ can be computed in $O(nk)$ time. \end{theorem}
\begin{proof}[Sketch] Correctness can be proved using Observation~\ref{obs:majObs}.  To bound the running time, note that topological sort takes time linear in the size of $D$, and the main loop (Lines~\ref{step:mainLoop}--\ref{step:endMainLoop}) examines each node and each edge once. Since the number of edges and nodes in $D$ is $O(nk)$,  the claimed time bound follows. \qed \end{proof}

The algorithm for strict consensus tree is similar to Algorithm~\ref{alg:Majority}, with only one significant difference:  instead of dealing with majority nodes, it focuses on \emph{strict} nodes, that is, TAG nodes $u$ such that $\cnt(u) = k$.  We omit the details, and simply summarize the result.

\begin{theorem} Given the TAG for a collection $\Prof$ of $k$ phylogenetic trees on the same $n$ leaves, the strict consensus tree of $\Prof$ can be computed in $O(nk)$ time. \end{theorem}

\section{Testing Compatibility using the TAG}
\label{compatibility}
Let $\mathcal{T}$ and $\mathcal{T}'$ be two phylogenetic trees on $X$ and $X'$, respectively, where $X \subseteq X'$. We say that $\mathcal{T}'$ \emph{displays} $\mathcal{T}$ if, up to suppressing non-root nodes of degree two, the minimum rooted subtree of $\mathcal{T}'$ that connects the elements of $X$ \emph{refines} $\mathcal{T}$, i.e., $\mathcal{T}$ can be obtained from it by contracting internal edges. \emph{Suppressing a node of degree two} means replacing that node and its incident edges by an edge.

Let $\Prof$ be the input collection of rooted phylogenetic trees.  We say that $\Prof$ is \emph{compatible} if there exists a phylogenetic tree $\mathcal{T}$, called a \emph{compatible supertree for $\T$}, that simultaneously displays every tree in $\Prof$.  A classic result in phylogenetics is that compatibility can be tested in polynomial time \cite{Aho:1981:ITL,Semple:2003:phy}.  In this section, we show that compatibility can be tested directly from the TAG for $\Prof$.

We need some definitions.  As before, we assume that multiple edges between the same two TAG vertices are replaced by a single edge. The \emph{extended TAG} is the graph $\DE$ obtained from $D$ by adding undirected edges between every two vertices $u,v \in U$ such that $f(u)$ and $f(v)$ are clusters corresponding to sibling vertices in some tree in $\Prof$.  $\DE$ is a  \emph{mixed} graph, i.e., a graph that contains both directed and undirected edges. See Fig. \ref{fig:comp}.

\begin{figure}
    \vspace{-0.2cm}
    \centering
     \begin{tikzpicture}
        \draw [thick,black] (1.5,1.3) to (0,0.7);
        \draw [thick,black] (1.5,1.3) to (1.5,-0.1) node [black,below] {\textbf{$c$}};
        \draw [thick,black] (0,0.7) to (0.4,-0.1) node [black,below] {\textbf{$b$}};
        \draw [thick,black] (0,0.7) to (-.5,-0.1) node [black,below] {\textbf{$a$}};
        \node at (0,1.4) {$\mathcal{T}_1$};

        \draw [thick,black] (4.5,1.3) to (3,0.7);
        \draw [thick,black] (4.5,1.3) to (4.5,-0.1) node [black,below] {\textbf{$d$}};
        \draw [thick,black] (3,0.7) to (3.4,-0.1) node [black,below] {\textbf{$b$}};
        \draw [thick,black] (3,0.7) to (2.5,-0.1) node [black,below] {\textbf{$a$}};
        \node at (3,1.4) {$\mathcal{T}_2$};

        \draw [thick,->] (7.1,0.2) to (6.5,-0.8);  
        \draw [thick,->] (7.1,0.2) to (7.7,-0.8);  
        \draw [thick,->] (8.1,1.4) to (7.2,0.4);  
        \draw [thick,->] (8.1,1.4) to (8.9,-0.8);  
        \draw [thick,dashed] (6.5,-1) to (7.7,-1);
        \draw [thick,dashed] (7.1,0.2) to (8.9,-1);
        \draw [thick,dashed] (7.1,0.2) to (10.1,-1);

        \draw [thick,->] (9.9,1.5) to (7.28,0.36);  
        \draw [thick,->] (9.9,1.5) to (10.1,-0.8);  

        \draw [thick,fill=green!90!black] (7.1,0.2) ellipse (3mm and 2mm);  
        \draw [thick,fill=green!90!black] (8.1,1.4) ellipse (3mm and 2mm);      
        \draw [thick,fill=green!90!black] (9.9,1.5) ellipse (3mm and 2mm);      

        \node at (7.4,1.6) {$\DE$};

        \draw [thick,fill=green!90!black] (6.5,-1) circle [radius=0.2];
        \draw [thick,fill=green!90!black] (7.7,-1) circle [radius=0.2];
        \draw [thick,fill=green!90!black] (8.9,-1) circle [radius=0.2];
        \draw [thick,fill=green!90!black] (10.1,-1) circle [radius=0.2];

        \node at (6.5,-1) {$a$};
        \node at (7.7,-1) {$b$};
        \node at (8.9,-1) {$c$};
        \node at (10.1,-1) {$d$};
     \end{tikzpicture}
    \vspace{-0.1cm}
    \caption{Two phylogenetic trees $\mathcal{T}_1$ and $\mathcal{T}_2$ with their extended TAG, $\DE$; undirected edges are shown with dashed lines. }
    \label{fig:comp}
    \vspace{-0.2cm}
  \end{figure}
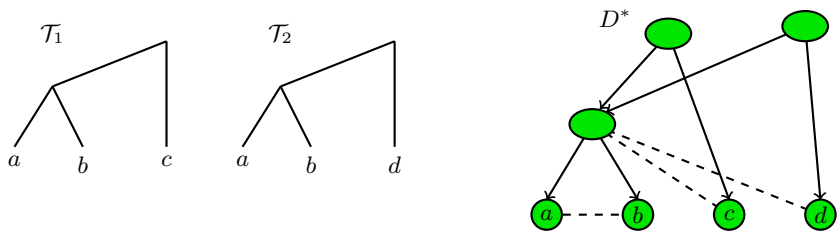

Let $D' = (U',E')$ be a mixed graph. An \emph{arc component} of  $D'$ is a maximal sub-mixed graph $W$ of $D'$ such that for every two nodes $u$ and $v$ in $W$ there is a path from $u$ to $v$ which consists only directed edges, irrespective of their directions. Let $v$ be a node of $D'$. The mixed graph obtained by deleting $v$ and its incident directed and undirected edges is denoted by $D' \setminus v$.  Let $V$ be a subset of $U'$.  We write $D'\setminus V$ to denote the (mixed) graph obtained from $D'$ by deleting each node in $V$ from $D'$.  The \emph{restriction of $D'$ to $V$}, denoted by $D'|V$, is the subgraph of $D'$ obtained by deleting each node in $U' \setminus V$ from $D'$.

\begin{figure}
\begin{algorithm}[H]
\SetAlgoLined
\SetNoFillComment
\DontPrintSemicolon
\KwIn{The extended TAG $\DE$ for a collection $\Prof$ of phylogenetic trees.}
\KwOut{A phylogenetic tree $\mathcal{T}$ that displays each tree in $\Prof$ or the statement \emph{not compatible}.}
Let $S_0$ be the set of nodes in $\DE$ that have in-degree zero and no incident edges. \;
\If{$S_0$ is empty}{\Return \emph{not compatible} \label{line:inc}}
\If{$S_0$ contains exactly one node with out-degree zero and label $\ell$}{
\Return the tree composed of singleton node with label $\ell$.}
Find the node sets $S_1, S_2, \dots ,S_m$ of the arc components of $\DE\setminus S_0$. \;
Delete all undirected edges of $\DE \setminus S_0$ whose endpoints are in distinct arc components. \;
\ForEach{$i \in \{1,2, \dots ,m\}$}{
	Call $\Desc(\DE|S_i)$ \;
	\If{this call returns \emph{not compatible}}{
		\Return  \emph{not compatible}}
	\Else{Let $\T_i$ be the phylogeny returned by this call}
}
\Return a tree with a root node and $\T_1, \T_2, \dots ,\T_m$ as its subtrees.
\caption{\Desc$(D\text{*})$\label{fig:algo2}}
\end{algorithm}
\vspace{-0.2cm}
\end{figure}

The extended TAG is closely related to the \emph{restricted descendancy graph} (RDG) \cite{Daniel:2005:AB,Berry:2006:RAB}. The RDG has a unique node for each internal input tree node along with $|S|$ leaf nodes. Let $u$ and $v$ be two input tree nodes, and let $u'$ and $v'$ be the corresponding nodes in the RDG.  If $u$ is a parent of $v$, then there is a directed edge from $u'$ to $v'$ in the RDG.  If $u$ and $v$ are siblings, then there is an undirected edge between $u'$ and $v'$ in the RDG. Otherwise, there is no edge between $u'$ and $v'$.

The extended TAG can be viewed as a compact version of the RDG of $\Prof$.  Thus, a slight adaptation of the \Desc\ algorithm  \cite{Daniel:2005:AB,Berry:2006:RAB} enables us to determine  whether $\Prof$ is compatible given its extended TAG $\DE$.  The details of this adaptation are shown in Algorithm \ref{fig:algo2}.  The algorithm first attempts to decompose the problem into subproblems, each of which corresponds to one of the subtrees of the compatible supertree.  If no such decomposition exists, then $\Prof$ is incompatible.  Otherwise, the algorithm identifies a collection of subproblems, each associated with a different arc component, and recursively tests compatibility for each subproblem.

\begin{theorem}\label{thm:compat}
Let $\DE$ be the extended TAG for a collection $\Prof$ of phylogenetic trees.  If $\Prof$ is compatible, then \Desc$(\DE)$ returns a compatible supertree for $\Prof$; otherwise,  \Desc$(\DE)$ returns the statement \emph{not compatible}. \end{theorem}
\begin{proof}[Sketch]  Let $\mathcal{P}'$ be the collection of phylogenetic trees after labeling the internal nodes of the input trees in $\mathcal{P}$ by their clusters. The order of labels in a cluster does not matter, that is, we assume two labels identical if their respective clusters are identical sets. Now the extended TAG of $\mathcal{P}$ is the same as the RDG of $\mathcal{P}'$. The correctness of Algorithm \ref{fig:algo2} now follows from the proof of \cite[Preposition 4]{Berry:2006:RAB} for $\mathcal{P}'$. We omit the details for lack of space. \qed \end{proof}

\paragraph{Running time:} Following \cite[Preposition 3]{Berry:2006:RAB}), we can show that if $\Prof$ consists of $k$ \emph{fully resolved} phylogenetic trees on the leaf set of size $n$, then the \Desc\ subroutine runs in time $O(n^2k^2)$.  We conjecture that the running time can be reduced to $O(nk\log^2n)$ using the approach discussed in \cite{Berry:2006:RAB}. If, however, the input trees are not fully resolved, the running time increases by a factor of $n$.

\paragraph{Remark.}  As we mention earlier, there are considerable similarities between the extended TAG and the RDG.  Nevertheless, the former has some advantages in practice. While every internal node of a tree in $\Prof$ gives rise to a distinct node in the RDG, the extended TAG has one node for each unique cluster. For instance, in the extreme case when $\Prof$ contains $k$ identical phylogenetic trees on $S$, the RDG has $O(nk)$ nodes, while the extended TAG contains only $O(n)$ nodes.  More typically, the trees in $\Prof$ will share many clusters, and the likelihood of this being the case is especially high when $k$ is much larger than $n$.

\section{Discussion}
\label{discussion}
We have presented a formal definition of the TAG that does not depend on the order of the input trees.
We have also presented a procedure for building TAGs, and described how to use TAGs to find consensus trees and to determine whether a collection of phylogenetic trees is compatible.

Extending TAGs to include potentially thousands of input trees from across the tree of life leads to several future challenges; two major ones are incorporating trees at different taxonomic levels and finding ways to synthesize conflicting phylogenetic input trees.  It may be possible to address the second of these challenges using ideas from the \textsc{AncestralBuild} algorithm \cite{Berry:2006:RAB,Daniel:2005:AB}.
Dealing with conflict among the input trees is also essential for processing large-scale phylogenetic data sets.  Although a visual inspection of the TAG provides some insight into the areas of conflict, approaches to quantify phylogenetic conflict within the TAG may provide valuable insight into mechanisms causing phylogenetic incongruence among biological datasets and help guide future phylogenetic research.
The synthesis approach of Smith et al. \cite{Smith:2015:OT} relies on a subjective ranking of the input trees. Potentially, a \textsc{MinCutSupertree} approach, like the \textsc{MultiLevelSupertree} algorithm \cite{Berry:2013:MLS}, could be applied to a TAG to provide an efficient and effective approach for synthesizing a tree of life.

\bibliographystyle{abbrv}

\end{document}